\documentclass[journal]{IEEEtran}                  

\usepackage{amsmath,amsfonts,amssymb,amsthm}
\usepackage{mathrsfs}
\usepackage{mathtools}
\usepackage{array}
\usepackage{float}
\usepackage{extarrows}
\usepackage{cite}
\usepackage{url}
\usepackage{xcolor}
\usepackage[colorlinks,linkcolor=blue,citecolor=blue]{hyperref}
\usepackage{algorithmic}
\usepackage{graphicx}
\usepackage{textcomp}
\usepackage[font=footnotesize]{caption}
\usepackage[font=footnotesize]{subcaption}
\usepackage{stfloats}
\usepackage{epstopdf}
\usepackage{arydshln}
\usepackage{ifthen}
\usepackage{pict2e}
\usepackage[T1]{fontenc}
\usepackage{enumitem}

\allowdisplaybreaks[1] 

\DeclareMathOperator{\diag}{diag}

\DeclareMathOperator{\trace}{trace}

\DeclareMathOperator{\vect}{vec}

\newtheorem{theorem}{Theorem}

\newcommand{\dd}{\mathrm{d}} 
\renewcommand{\t}{^{\mbox{\tiny\sf T}}} 

\newcommand{\R}{\mathbb{R}}

\def\send#1#2{\stackrel{#1}{\hbox to #2{\rightarrowfill}}}
\def\-{\!\!\!\!\!-}

\newcommand{\rank}{{\rm rank\;}}

\newtheorem{lemma}{Lemma}
\newtheorem{remark}{Remark}

\newtheorem{corollary}{Corollary}

\def\R{{\rm I\!R}} 

\newcounter{seqn}[equation]
\def\theseqn{\arabic{equation}\alph{seqn}}

\def\endseqn{\eqno \@seqnnum
$$\ignorespaces}
\def\@seqnnum{(\theseqn)}

\newskip\mcentering \mcentering=0pt plus 1000pt minus 1000pt

\def\meqalignno#1{
\halign to\displaywidth{
    \hbox to 0pt{\kern\displaywidth\llap{$##$}\hss}\tabskip=\mcentering
    &\hfil$\displaystyle{##}$\tabskip=\mcentering
   &&$\displaystyle{{}##}$\hfil\tabskip=\mcentering
    \crcr
    #1\crcr}}

\def\dspace{\multiply\normalbaselineskip 150
		  \divide\normalbaselineskip 100 \normalbaselines
		  \csname @@normalbaselineskip\endcsname\normalbaselineskip}
\def\sspace{\multiply\normalbaselineskip 200
		 \divide\normalbaselineskip 300 \normalbaselines
		 \csname @@normalbaselineskip\endcsname\normalbaselineskip}
\def\sdspace{\multiply\normalbaselineskip 160
		 \divide\normalbaselineskip 150 \normalbaselines
		 \csname @@normalbaselineskip\endcsname\normalbaselineskip}

\def\@{\tilde}

\def\3dot#1{\buildrel\textstyle...\over#1}

\renewcommand{\R}{\mathbb{R}}

\begin{document}

\title{Optimal Covariance Steering for \\
Continuous-Time Linear Stochastic Systems \\
With Additive Noise}

\author{Fengjiao Liu and Panagiotis Tsiotras
\thanks{This work has been supported by NASA University Leadership Initiative award 80NSSC20M0163 and ONR award N00014-18-1-2828.}
\thanks{F. Liu and P. Tsiotras are with the School of Aerospace Engineering, Georgia Institute of Technology, Atlanta, GA 30332 USA (e-mail: \{fengjiao, tsiotras\}@gatech.edu).}}

\maketitle

\begin{abstract}
In this paper, we study the problem of how to optimally steer the state covariance of a general continuous-time linear stochastic system over a finite time interval subject to additive noise. 
Optimality here means reaching a target state covariance with minimal control energy. 
The additive noise may include a combination of white Gaussian noise and abrupt ``jump noise'' that is discontinuous in time. 
We first establish the controllability of the state covariance for linear time-varying stochastic systems. 
We then turn to the derivation of the optimal control, which entails solving two dynamically coupled matrix ordinary differential equations (ODEs) with split boundary conditions. 
We show the existence and uniqueness of the solution to these coupled matrix ODEs, and thus those of the optimal control.
\end{abstract}

\section{Introduction}

All dynamical systems are prone to disturbances whose effects persist with time. 
Controlling uncertainty is critical for the robustness and overall performance of all such systems. 
Among the various approaches, covariance control theory, which has been developed since the mid-1980s, provides a \textit{direct} way to regulate the state covariance of a stochastic system subject to noise~\cite{hotz1985covariance, collins1985covariance}. 
By quantifying uncertainty in a direct manner, covariance control theory has found applications in many real-life engineering and physics problems, such as spacecraft soft landing \cite{ridderhof2018uncertainty} and trajectory optimization \cite{ridderhof2020fuel}, vehicle path planning \cite{okamoto2019optimal, zheng2022belief, yin2022trajectory}, drone delivery, multi-agent systems \cite{saravanos2021distributed}, active cooling of stochastic oscillators \cite{chen2021optimal}, Schr\"{o}dinger bridge problem \cite{chen2016stochastic}, and optimal mass transport problem \cite{chen2016relation}.

Much of the earlier work dealt exclusively with controlling the stationary, or asymptotic, state covariance over an infinite time horizon \cite{hotz1987covariance, collins1987theory, yasuda1993covariance, georgiou2002structure, zhu1997convergent}. 
Most recent work has focused on the finite-horizon optimal covariance steering problem in both continuous-time \cite{chen2016I, chen2016II, chen2018III, ciccone2020regularized} and discrete-time settings \cite{bakolas2018finite, balci2021covariance, balci2021Wasserstein, okamoto2018optimal, sivaramakrishnan2021distribution, pilipovsky2021covariance}. 
It is shown that when the noise channel coincides with the control channel, and the noise is modeled by a Wiener process, there exists a unique optimal control that steers the state covariance from any initial positive definite matrix $\Sigma_{0} \succ 0$ to any final $\Sigma_{T} \succ 0$ over a finite time interval $[0, T]$ \cite{chen2016I, chen2018III}. 
The cases when one or both of the initial and final state covariances $\Sigma_{0}$ and $\Sigma_{T}$ are singular are treated in \cite{ciccone2020regularized}. 
When the noise channel differs from the control channel, and the coefficient matrices $A$, $B$, and $C$ are constant, controllability of the state covariance via state feedback is established in \cite{chen2016II}. 
However, the question of the existence and uniqueness of an optimal control policy is still open for different noise and control channels and more comprehensive noise models.


\textit{Contributions:} 
In this paper, 
we first extend the noise model used in covariance steering problems to include jumps of any size. 
Next, we establish a controllability result for the covariance equation of a linear time-varying stochastic system under mild assumptions. 
Lastly and most importantly, we show the existence and uniqueness of the optimal control law when the noise and control channels are different.


The rest of the paper is organized as follows. 
The main problem of interest is formulated in Section \ref{sec:prob}. 
The comprehensive noise model is discussed in Section \ref{sec:noise}. 
The controllability of the state covariance for a linear time-varying stochastic system is shown in Section \ref{sec:contr}. 
The existence and uniqueness of the optimal control are shown in Section \ref{sec:opt-ctrl}. 
Finally, two examples are presented in Section \ref{sec:expl} to illustrate the results of this paper. 
Some of the longer proofs are provided in the Appendix.

\section{Problem Formulation} \label{sec:prob}

Consider the following time-varying linear stochastic system subject to additive martingale noise
\begin{align} \label{sde:add-jump}
&\dd x(t) = A(t) x(t) \, \dd t + B(t) u(t) \, \dd t + C(t) \, \dd m(t), \\
&\mathbb{E}\left[x(0)\right] = 0, \quad \mathbb{E}\big[x(0) x\t(0) \big] = \Sigma_{0} \succ 0, \label{bgn:mean-cov}
\end{align}
where $x(t) \in \R^{n}$ is the state, $u(t) \in \R^{p}$ is the control input, 
$m(t) \in \R^{q}$ is a martingale \cite{cinlar2011stochastics} independent of $x(0)$ with $\dd\mathbb{E}\big[m(t) m\t(t) \big]/\dd t = D(t) \succeq 0$, 
and $A(t) \in \R^{n \times n}$, $B(t) \in \R^{n \times p}$, and $C(t) \in \R^{n \times q}$ are the coefficient matrices. 
Let $\mathcal{C}^{k}$ denote the class of $k$-times continuously differentiable functions. 
We assume that $A(t) \in \mathcal{C}^{n-1}$, $B(t) \in \mathcal{C}^{n}$, and $C(t), D(t) \in \mathcal{C}^{0}$. 
Without loss of generality, 
assume that \eqref{sde:add-jump} is defined on the time interval $[0, 1]$, and that the desired terminal state $x(1)$ is characterized by its mean\footnote{Assuming zero mean is just for convenience, otherwise the control will have an additional term; see Remark~9 in \cite{chen2016I}.} and covariance matrix given by 
\begin{equation} \label{tgt:mean-cov}
\mathbb{E}\left[x(1)\right]=0,
\quad
\mathbb{E}\big[x(1) x\t(1) \big]=\Sigma_{1} \succ 0.
\end{equation}
A control input $u \in \mathcal{U}$ is \emph{admissible} if,  for each $t \in [0, 1]$, 
it depends only on $t$ and on the past history of the states $\{x(s):\, 0 \leq s \leq t\}$, and satisfies
\begin{equation} \label{cost}
J(u) \triangleq \mathbb{E}\left[\int_{0}^{1} u\t(t) R(t) u(t) \, \dd t\right] < \infty,
\end{equation}
for some given continuous $R(t) \succ 0$ of dimension $p \times p$ for all $t \in [0, 1]$, such that \eqref{sde:add-jump} with the initial condition \eqref{bgn:mean-cov} has a strong solution \cite{protter2003stochastic}, and the desired terminal state mean and covariance given by \eqref{tgt:mean-cov} are achieved.

The problem is to determine whether $\mathcal{U}$ is nonempty and, if so, to find the optimal control $u^{*} \in \mathcal{U}$ that minimizes the quadratic cost functional \eqref{cost} subject to the boundary constraints \eqref{bgn:mean-cov}, \eqref{tgt:mean-cov}.

\section{Martingale Noise Model} \label{sec:noise}

In this section, we give some examples of the martingale noise used in this work, along with a brief discussion on how the martingale noise affects controlling the system.

The martingale noise $\dd m(t)$ has a mean value of zero. 
Generally speaking, the noise may not be white, since a martingale may not have stationary and independent increments. 
Some examples of noise modeled by $\dd m$ are as follows.
\begin{enumerate}[leftmargin=*]

\item
\textbf{White Gaussian noise:} $\dd m_{1}(t) = \dd w(t)$, where $w(t) \in \R^{q}$ is a Wiener process.

\item
\textbf{Noise with jumps:} $\dd m_{2}(t)= \dd h(t)-\lambda(t) \, \dd t$, where $h(t) \in \R^{q}$ is a nonhomogeneous Poisson process with (deterministic) arrival rate $\lambda(t) = [\lambda_{1}(t) \enspace \lambda_{2}(t) \enspace \dots \enspace \lambda_{q}(t)]\t$, since the compensated nonhomogeneous Poisson process $m_{2}(t) = h(t) - \int_{0}^{t} \lambda(\tau) \, \dd \tau$ is a martingale \cite{cinlar2011stochastics}. 
In general, the jump size may not be fixed but it is a random variable following a certain distribution. 
We may take $m_2$ to be a nonhomogeneous compound Poisson process subtracting its compensator \cite{cinlar2011stochastics}. 
Thus, $\dd m_2$ can model random noise of any jump size.

\item
\textbf{Combinations:} We can model continuous noise and noise with jumps, such as $\dd m_{3}(t) = \kappa \, \dd w(t) + \dd m_{2}(t)$, where $\kappa>0$, and $\dd m_{4}(t) = \dd \ell(t) - \dd e(t)$, where $\ell$ is a L\'{e}vy process and $e$ is its compensator \cite{cinlar2011stochastics}. 

\end{enumerate}

As the martingale noise is characterized solely by its mean and covariance, it is not surprising that the optimal control for martingale noise has the same structure as that for Gaussian noise. 
In both cases, the control boils down to solving two coupled matrix ordinary differential equations (ODEs), one Riccati equation and one Lyapunov equation, with split boundary conditions on the latter. 
In Section \ref{sec:opt-ctrl}, we provide a new, unifying approach to analyze these coupled matrix ODEs.

\section{Controllability of the State Covariance} \label{sec:contr}

In this section, the controllability of the state covariance for \eqref{sde:add-jump} is established. 
The state covariance of \eqref{sde:add-jump}, denoted by
\begin{equation*}
\Sigma(t) \triangleq \mathbb{E}\left[\big(x(t)-\mathbb{E}[x(t)]\big)\big(x(t)-\mathbb{E}[x(t)]\big)\t \right], 
\end{equation*}
is said to be \emph{controllable} on the time interval $[0, 1]$ if, for any given $\Sigma_{0}, \Sigma_{1} \succ 0$, there exists an admissible control $u \in \mathcal{U}$ that steers the state covariance $\Sigma(t)$ of \eqref{sde:add-jump} from $\Sigma(0) = \Sigma_{0}$ to $\Sigma(1) = \Sigma_{1}$, while maintaining $\Sigma(t) \succ 0$ for all $t \in [0, 1]$.

Consider the state feedback control law of the form 
\begin{equation} \label{eqn:stt-fb}
u(t) = K(t)x(t), \quad t \in [0, 1], 
\end{equation} 
where $K(t) \in \R^{p \times n}$ is the feedback matrix to be determined. 
Since $\mathbb{E}\left[x(0)\right] = 0$, we have $\mathbb{E}\left[x(t)\right] = 0$ for all $t \in [0, 1]$. 
Therefore, the state covariance $\Sigma(t) = \mathbb{E}\big[x(t) x\t(t) \big]$ satisfies
\begin{multline} \label{ode:sigma-stt-fb}
\dot{\Sigma} = \big(A(t)+B(t)K(t)\big)\Sigma + \Sigma\big(A(t)+B(t)K(t)\big)\t 
\\* 
+ C(t) D(t) C\t(t).
\end{multline}
Note that with the control \eqref{eqn:stt-fb}, controlling the covariance of the stochastic system \eqref{sde:add-jump} amounts to controlling the deterministic system \eqref{ode:sigma-stt-fb}. 
When the matrices $A$, $B$, $C$, and $D$ are constant and the pair $(A, B)$ is controllable, the state covariance dynamics \eqref{ode:sigma-stt-fb} is controllable via $K(t)$ \cite{chen2016II}, since \eqref{ode:sigma-stt-fb} can be reduced to the form of (27) in \cite[Theorem 3]{chen2016II}. 
It is not difficult to check that this statement still holds when $C(t)$ and $D(t)$ are time-varying. 
It will be shown below that when $A(t)$ and $B(t)$ are  time-varying, and under some mild assumptions, \eqref{ode:sigma-stt-fb} is also controllable and therefore $\mathcal{U}$ is non-empty.

Given a time-varying matrix pair $\big(A(t), B(t)\big)$ of dimensions $n \times n$ and $n \times p$, respectively, define 
\begin{align*}
\Theta_{i}(t) &\triangleq 
\begin{bmatrix}
\Gamma_{0}(t) & \Gamma_{1}(t) & \cdots & \Gamma_{i-1}(t)
\end{bmatrix}
, \enspace 1 \leq i \leq n+1, \\
\Gamma_{0}(t) &\triangleq B(t)
, \\ 
\Gamma_{k}(t) &\triangleq -A(t)\Gamma_{k-1}(t) + \dot{\Gamma}_{k-1}(t)
, \quad 1 \leq k \leq n.
\end{align*}
Then, the \emph{controllability matrix} of $\big(A(t), B(t)\big)$ is $\Theta_{n}(t)$  and the pair $\big(A(t), B(t)\big)$ is \emph{uniformly controllable} on the time interval $[0, 1]$ if, for all $t \in [0, 1]$, $\rank \Theta_{n}(t) = n$ \cite{silverman1967controllability}. 
Recall that the controllability property in \cite{chen2016I} requires that the controllability Gramian
\begin{equation} \label{eqn:N-def}
N(t_{1}, t_{0}) \triangleq \int_{t_{0}}^{t_{1}} \Phi_{A}(t_{0}, \tau) B(\tau) B\t(\tau) \Phi_{A}(t_{0}, \tau)\t \, \dd\tau
\end{equation}
is nonsingular for all $0 \leq t_{0} < t_{1} \leq 1$, where $\Phi_{A}(t, \tau)$ is the state transition matrix of $A(t)$ (see also \eqref{pde:stt-trans-A} below). 
This property is, in fact, equivalent to the statement that the pair $\big(A(t), B(t)\big)$ is \emph{totally controllable} on the time interval $[0, 1]$ \cite{kreindler1964concepts}, which holds if and only if, for all $0 \leq t_{0} < t_{1} \leq 1$, there exists $t \in (t_{0}, t_{1})$ such that $\rank \Theta_{n}(t) = n$ \cite{stubberud1964controllability}. 
Clearly, uniform controllability is a stronger property than total controllability. 
The pair $\big(A(t), B(t)\big)$ is \emph{index invariant} on the interval $[0, 1]$ if, for each $i \in \{1, 2, \dots, n+1\}$, $\rank \Theta_{i}(t)$ is constant for all $t \in [0, 1]$, and $\rank \Theta_{n}(t) = \rank \Theta_{n+1}(t)$ \cite{morse1973structure}.


\begin{theorem} \label{thm:contr}
Let the pair $\big(A(t), B(t)\big)$ be uniformly controllable and index invariant on the time interval $[0, 1]$. 
Then, the state covariance of the linear stochastic system \eqref{sde:add-jump} is controllable on $[0, 1]$ and also on any subinterval of $[0, 1]$. 
\end{theorem}

\begin{proof}
Without loss of generality, assume $B(t)$ is of full column rank for all $t \in [0, 1]$. 
Since $\big(A(t), B(t)\big)$ is uniformly controllable and index invariant, there exists a time-varying coordinate transformation that brings the pair $\big(A(t), B(t)\big)$ into its canonical form \cite{wolovich1968stabilization, seal1969canonical}. 
That is, there exist nonsingular matrices $P(t) \in \R^{n \times n}$ and $Q(t) \in \R^{p \times p}$ such that the new state $\tilde{x}(t) = P(t)x(t)$ satisfies
\begin{equation*}
\dd\tilde{x}(t) = \tilde{A}(t) \tilde{x}(t) \, \dd t + \tilde{B} \tilde{u}(t) \, \dd t + \tilde{C}(t) \, \dd m(t),
\end{equation*}
where $\tilde{A}(t) = \big(P(t)A(t) + \dot{P}(t)\big) P^{-1}(t)$, $\tilde{B} = P(t)B(t)Q(t)$, $\tilde{C}(t) = P(t)C(t)$, $\tilde{u}(t) = Q^{-1}(t) u(t)$, and the pair $\big(\tilde{A}(t), \tilde{B}\big)$ is in the canonical form
\begin{align*}
&\tilde{A}(t) = 
	\begin{bmatrix}
	\tilde{A}_{11} & \tilde{A}_{12} & \dots & \tilde{A}_{1p} \\
	\tilde{A}_{21} & \tilde{A}_{22} & \dots & \tilde{A}_{2p} \\
	\vdots & \vdots & \ddots & \vdots \\
	\tilde{A}_{p1} & \tilde{A}_{p2} & \dots & \tilde{A}_{pp}
	\end{bmatrix}
	, \enspace
\tilde{B} = 
	\begin{bmatrix}
	b_{1} & 0 & \dots & 0 \\
	0 & b_{2} & \dots & 0 \\
	\vdots & \vdots & \ddots & \vdots \\
	0 & 0 & \dots & b_{p}
	\end{bmatrix}
\\
&\tilde{A}_{ii} = 
	\begin{bmatrix}
	0 & 1 & \dots & 0 \\
	\vdots & \ddots & \ddots & \vdots \\
	0 &  \dots & 0 & 1 \\
	\times & \dots & \times & \times
	\end{bmatrix}
	, \,
\tilde{A}_{ij} = 
	\begin{bmatrix}
	0 & \dots & 0 \\
	\vdots & \vdots & \vdots \\
	0 & \dots & 0 \\
	\times & \dots & \times
	\end{bmatrix}
	, \,
b_{i} = 
	\begin{bmatrix}
	0 \\
	\vdots \\
	0 \\
	1
	\end{bmatrix}
\\
&\tilde{A}_{ij}(t) \in \R^{n_{i} \times n_{j}}, \enspace 
i, j \in \{1, 2, \dots, p\}, \enspace 
n_{1} + \dots + n_{p} = n.
\end{align*}
With the state feedback control $\tilde{u}(t) = \tilde{K}(t) \tilde{x}(t)$, where $\tilde{K}(t) = Q^{-1}(t) K(t) P^{-1}(t)$, 
the new state covariance 
\begin{align*}
\tilde{\Sigma}(t) 
&\triangleq \mathbb{E}\left[\big(\tilde{x}(t)-\mathbb{E}[\tilde{x}(t)]\big) \big(\tilde{x}(t)-\mathbb{E}[\tilde{x}(t)]\big)\t \right] \\
&= P(t) \Sigma(t) P\t(t)
\end{align*}
satisfies
\begin{multline*}
\dot{\tilde{\Sigma}} = \big(\tilde{A}(t) + \tilde{B} \tilde{K}(t)\big) \tilde{\Sigma} + \tilde{\Sigma} \big(\tilde{A}(t) + \tilde{B} \tilde{K}(t)\big)\t 
\\* 
+ \tilde{C}(t)D(t)\tilde{C}\t(t). 
\end{multline*}
From the previous expressions of $\tilde{A}(t)$ and $\tilde{B}$, 
it follows that there exists 
\begin{equation*}
\tilde{F}(t) = - \big[\tilde{a}_{1}\t(t) \quad \tilde{a}_{2}\t(t) \quad \dots \quad \tilde{a}_{p}\t(t) \big]\t \in \R^{p \times n},
\end{equation*}
where $\tilde{a}_{i}(t) \in \R^{1 \times n}$ is the $\big(\sum_{j=1}^{i} n_{j}\big)$th row of $\tilde{A}(t)$ for $i \in \{1, \dots, p\}$, such that $\big(\tilde{A}(t)+\tilde{B}\tilde{F}(t), \tilde{B}\big)$ is a time-invariant matrix pair. 
Notice that the initial and terminal constraints \eqref{bgn:mean-cov} and \eqref{tgt:mean-cov} in the new coordinates are 
\begin{align*}
&\mathbb{E}\left[\tilde{x}(0)\right] = 0,
\quad 
\mathbb{E}\left[ \tilde{x}(0) \tilde{x}\t(0) \right] = \tilde{\Sigma}_{0} = P(0)\Sigma_{0} P\t(0) \succ 0, \nonumber \\
&\mathbb{E}\left[\tilde{x}(1)\right]=0,
\quad
\mathbb{E}\left[ \tilde{x}(1) \tilde{x}\t(1) \right] = \tilde{\Sigma}_{1} = P(1)\Sigma_{1} P\t(1) \succ 0.
\end{align*}
Clearly, if $\tilde{\Sigma}(t) \succ 0$ for all $t \in [0, 1]$, then $\Sigma(t) \succ 0$ for all $t \in [0, 1]$. 
Therefore, the controllability of the state covariance for the time-varying system \eqref{sde:add-jump} follows from that for the time-invariant system $\big(\tilde{A}(t)+\tilde{B}\tilde{F}(t), \tilde{B}\big)$ \cite{chen2016II}. 
Since $\big(A(t), B(t)\big)$ is uniformly controllable and index invariant on $[0, 1]$, it follows that $\big(A(t), B(t)\big)$ is also uniformly controllable and index invariant on any subinterval. \hfill
\end{proof}

\begin{remark}
Theorem \ref{thm:exist-unique} shows that if the pair $\big(A(t), B(t)\big)$ is totally controllable on $[0, 1]$, there exists a unique optimal control for any $\Sigma_{0}, \Sigma_{1} \succ 0$. 
It follows that if the pair $\big(A(t), B(t)\big)$ is totally controllable on $[0, 1]$, the state covariance of \eqref{sde:add-jump} is controllable on $[0, 1]$ and also on any subinterval of $[0, 1]$.
\end{remark}

\section{Optimal Control of the State Covariance} \label{sec:opt-ctrl}

In this section, the existence and uniqueness of the optimal control for the covariance steering problem are established. 
It is assumed throughout this section that the pair $\big(A(t), B(t)\big)$ is totally controllable on $[0, 1]$.


Using the ``completion of squares'' argument \cite{chen2016I, chen2016II, chen2018III}, a candidate optimal control law can be derived as follows.

\begin{lemma} \label{lem:suff}
Assume $\Pi(t)$ and $\Sigma(t)$ satisfy, for $t \in [0, 1]$, 
\begin{align} 
&\dot{\Pi} = - A\t(t) \Pi - \Pi A(t) + \Pi B(t) R^{-1}(t) B\t(t) \Pi,
\label{ode:pi} \\
&\dot{\Sigma} = \big(A(t) - B(t) R^{-1}(t) B\t(t) \Pi(t)\big)\Sigma + C(t) D(t) C\t(t)
\nonumber \\* 
&\hspace{16mm} 
+ \Sigma\big(A(t) - B(t) R^{-1}(t) B\t(t) \Pi(t)\big)\t,
\label{ode:sigma} \\
&\Sigma(0) = \Sigma_{0} \succ 0, \qquad \Sigma(1) = \Sigma_{1} \succ 0.
\label{bdr:sigma}
\end{align}
Then, the state feedback control 
\begin{equation} \label{ctrl:opt}
u^{*}(t) = - R^{-1}(t) B\t(t) \Pi(t) x(t)
\end{equation}
is optimal with respect to the cost functional \eqref{cost}, subject to the boundary constraints \eqref{bgn:mean-cov}, \eqref{tgt:mean-cov}. 
\end{lemma}


Next, we show the following main result by analyzing the solution to the coupled matrix ODEs \eqref{ode:pi}, \eqref{ode:sigma}, \eqref{bdr:sigma}.

\begin{theorem} \label{thm:exist-unique}
For any given $\Sigma_{0}, \Sigma_{1} \succ 0$, the unique optimal control that solves the covariance steering problem is given by \eqref{ctrl:opt}, where $\Pi(t)$ is the unique solution to \eqref{ode:pi}, \eqref{ode:sigma}, \eqref{bdr:sigma}.
\end{theorem}

First, it is straightforward to check that, for $R(t) \succ 0$, if the pair $\big(A(t), B(t)\big)$ is totally controllable on $[0, 1]$, then the pair $\big(A(t), B(t)R^{-\frac{1}{2}}(t)\big)$ is also totally controllable on $[0, 1]$. 
As $R(t)$ can always be recovered by a time-varying coordinate transformation, for simplicity, we assume $R(t) \equiv I_{p}$ for the rest of this section. 
To begin with, a complete solution of $\Pi(t)$ to \eqref{ode:pi} is presented.

\subsection{Solution to Riccati Differential Equation}

Let $\Phi_{A}(t, s)$ denote the state transition matrix of $A(t)$, which satisfies 
\begin{equation} \label{pde:stt-trans-A}
\frac{\partial}{\partial t} \Phi_{A}(t, s) = A(t) \Phi_{A}(t, s), \quad \Phi_{A}(s, s) = I_{n}.
\end{equation}
Let the matrix $N(t,s)$ as in \eqref{eqn:N-def}. 
Since $\big(A(t), B(t)\big)$ is totally controllable on $[0, 1]$, $N(t, s) \succ 0$ for all $0 \leq s < t \leq 1$.


\begin{lemma} \label{lem:pi-exist-sol}
Given $\Pi(s)$ for some $s \in [0, 1]$, \eqref{ode:pi} admits a unique solution $\Pi(t)$ on $[0, 1]$ if and only if
\begin{equation} \label{cond:pi-exist}
N(0, s)^{-1} \prec \Pi(s) \prec N(1, s)^{-1},
\end{equation}
where $N(0, 0^{+})^{-1} = - \infty$ and $N(1, 1^{-})^{-1} = + \infty$\footnote{Positive infinity of the $n \times n$ positive semidefinite cone, written $+ \infty$, is the limit of a sequence of $n \times n$ positive definite matrices whose eigenvalues all grow to $+ \infty$. 
Likewise, for $- \infty$. 
Notice that as $s \to 0^{+}$, all eigenvalues of $N(0, s)$ go to $0^{-}$, and thus all eigenvalues of $N(0, s)^{-1}$ go to $- \infty$. 
Likewise, as $s \to 1^{-}$, all eigenvalues of $N(1, s)$ go to $0^{+}$, and thus all eigenvalues of $N(1, s)^{-1}$ go to $+ \infty$.}. 
Moreover, 
\begin{align} \label{sol:pi}
\Pi(t) &= \Phi_{A}(s, t)\t \Pi(s) \big(I_{n} - N(t, s) \Pi(s)\big)^{-1} \Phi_{A}(s, t) \\
&= \Phi_{A}(s, t)\t \big(I_{n} - \Pi(s) N(t, s)\big)^{-1} \Pi(s) \Phi_{A}(s, t), \nonumber
\end{align}
and 
\begin{equation} \label{ineq:pi-bd}
N(0, t)^{-1} \prec \Pi(t) \prec N(1, t)^{-1}, \quad t \in [0, 1].
\end{equation}
\end{lemma}

\begin{proof}
First, recall that \eqref{ode:pi} admits a unique solution $\Pi(t)$ on $[0, 1]$ if and only if $\Phi_{A}(t, s) - \Phi_{A}(t, s) N(t, s) \Pi(s)$ is invertible for all $t \in [0, 1]$, and the solution $\Pi(t)$ is given by \eqref{sol:pi}~\cite{kilicaslan2010existence}.

Next, we show \eqref{cond:pi-exist}. 
Since $\Phi_{A}(t, s)$ is invertible, it suffices to show that $I_{n} - N(t, s) \Pi(s)$ is invertible for all $t \in [0, 1]$ if and only if \eqref{cond:pi-exist} holds.

(Necessity) Let $I_{n} - N(t, s) \Pi(s)$ be invertible for all $t \in [0, 1]$. 
Then, $M(t) \triangleq I_{n} - N(t, s)^{\frac{1}{2}} \Pi(s) N(t, s)^{\frac{1}{2}}$ is invertible for all $t \in [s, 1]$. 
Since $M(s) = I_{n} \succ 0$ and $M(t)$ is continuous in $t \in [s, 1]$, it follows that $M(t) \succ 0$ for all $t \in [s, 1]$. 
In particular, $M(1) \succ 0$, that is, $I_{n} \succ N(1, s)^{\frac{1}{2}} \Pi(s) N(1, s)^{\frac{1}{2}}$. 
Hence, $\Pi(s) \prec N(1, s)^{-1}$. 
The inequality $\Pi(s) \succ N(0, s)^{-1}$ can be shown in an analogous way, and thus is omitted.

(Sufficiency) Let \eqref{cond:pi-exist} hold. 
We will show that $I_{n} - N(t, s) \Pi(s)$ is invertible for all $t \in [0, 1]$. 
When $t = s$, we have $I_{n} - N(s, s) \Pi(s) = I_{n}$, which is invertible. 
When $t \in (s, 1]$, and since $0 \prec N(t, s) \preceq N(1, s)$, it follows from \eqref{cond:pi-exist} that $N(t, s)^{-1} \succeq N(1, s)^{-1} \succ \Pi(s)$. 
Hence, by multiplying $N(t, s)^{\frac{1}{2}}$ on both sides of the above inequality, we obtain $I_{n} \succ N(t, s)^{\frac{1}{2}} \Pi(s) N(t, s)^{\frac{1}{2}}$. 
Since $N(t, s)^{\frac{1}{2}} \Pi(s) N(t, s)^{\frac{1}{2}}$ is similar to $N(t, s) \Pi(s)$, all the eigenvalues of $N(t, s) \Pi(s)$ are strictly less than $1$. 
Therefore, $I_{n} - N(t, s) \Pi(s)$ is invertible for all $t \in (s, 1]$. 
The case when $t \in [0, s)$ can be shown in an analogous way, and thus is omitted.

Next, we show \eqref{ineq:pi-bd}. Notice that, for any fixed $r$, $N(r, t)^{-1}$ satisfies the same Riccati differential equation as $\Pi(t)$, that is,
\begin{align} \label{ode:RiccN}
\frac{\partial}{\partial t} N(r, t)^{-1} &= - A\t(t) N(r, t)^{-1} - N(r, t)^{-1} A(t) \nonumber \\* 
&\qquad + N(r, t)^{-1} B(t) B\t(t) N(r, t)^{-1}.
\end{align}
The monotonicity of the matrix Riccati differential equation \cite{freiling1996generalized} implies that, 
for any $s \in [0, 1]$ such that $N(0, s)^{-1} \prec \Pi(s) \prec N(1, s)^{-1}$,
it follows that, for all $t \in [0, 1]$, $N(0, t)^{-1} \prec \Pi(t) \prec N(1, t)^{-1}$.  \hfill
\end{proof}


\begin{corollary} \label{cor:pi-exist}
Assume that $\big(A(t), B(t)\big)$ is totally controllable for $t \in \R$. 
Let $\mathcal{I}_{s} \subset \R$ be the maximal interval of existence of the solution to \eqref{ode:pi}, starting from $\Pi(s) = \Pi_{s}$. Then, $\mathcal{I}_{s} = (t_{0}, t_{1})$, where
\begin{align*}
t_{0} &\triangleq \inf \big\{t \, \big| \, t < s,~ N(t, s)^{-1} \prec \Pi_{s}\big\}, \\
t_{1} &\triangleq \sup \big\{t \, \big| \, t > s,~ N(t, s)^{-1} \succ \Pi_{s}\big\}.
\end{align*}
\end{corollary}

\subsection{Solution to the State Covariance}

Armed with the complete solution of $\Pi(t)$, we obtain an explicit expression for the state transition matrix $\Phi_{A-BB\t \Pi}(t, s)$, which satisfies $\Phi_{A-BB\t\Pi}(s, s) = I_{n}$, and 
\begin{equation*}
\frac{\partial}{\partial t} \Phi_{A-BB\t\Pi}(t, s) = \hspace{-0.4mm} \big(A(t) - B(t) B\t(t) \Pi(t)\big) \Phi_{A-BB\t\Pi}(t, s).
\end{equation*}

\begin{lemma} \label{lem:phi-AB-pi}
Let condition \eqref{cond:pi-exist} hold, so that $\Pi(t)$ exists on $[0, 1]$. 
The state transition matrix of $A(t)-B(t) B\t(t) \Pi(t)$ is given, for $s, t \in [0, 1]$, by
\begin{equation} \label{eqn:phi-AB-pi}
\Phi_{A-BB\t\Pi}(t, s) = \Phi_{A}(t, s) - \Phi_{A}(t, s) N(t, s) \Pi(s). 
\end{equation}
\end{lemma}

\begin{proof}
One can readily check that
\begin{equation*}
\Phi_{A}(s, s) - \Phi_{A}(s, s) N(s, s) \Pi(s) = I_{n}.
\end{equation*}
Since,
\begin{equation*}
\frac{\partial}{\partial t} N(t, s) = \Phi_{A}(s, t) B(t) B\t(t) \Phi_{A}(s, t)\t,
\end{equation*}
in view of \eqref{sol:pi}, we obtain that
\begin{align*}
&\frac{\partial}{\partial t} \big(\Phi_{A}(t, s) - \Phi_{A}(t, s) N(t, s) \Pi(s)\big) \\
&= \big(A(t) - B(t) B\t(t) \Pi(t)\big) \\* 
&\qquad\qquad \times \big(\Phi_{A}(t, s) - \Phi_{A}(t, s) N(t, s) \Pi(s)\big).
\end{align*}
This completes the proof. \hfill 
\end{proof}

\begin{remark}
If \eqref{cond:pi-exist} holds, in view of \eqref{sol:pi} and \eqref{eqn:phi-AB-pi}, we have
\begin{align*}
\Pi(t) &= \Phi_{A}(s, t)\t \Pi(s) \Phi_{A-BB\t\Pi}(s, t) \\
&= \Phi_{A-BB\t\Pi}(s, t)\t \Pi(s) \Phi_{A}(s, t), \quad s, t \in [0, 1].
\end{align*}
\end{remark}


The next result provides an explicit solution of $\Sigma(t)$ over the interval $[0, 1]$.

\begin{lemma} \label{lem:sigma1}
Let condition \eqref{cond:pi-exist} hold so that $\Pi(t)$ exists on $[0, 1]$. Given $\Sigma(0) = \Sigma_{0} \succeq 0$, then, for $t \in [0, 1]$,
\begin{align} \label{sol:sigma}
&\Sigma(t) = \Phi_{A-BB\t\Pi}(t, 0) \Sigma_{0} \Phi_{A-BB\t\Pi}(t, 0)\t 
\nonumber \\* 
&+ \int_{0}^{t} \Phi_{A-BB\t\Pi}(t, s) C(s) D(s) C\t(s) \Phi_{A-BB\t\Pi}(t, s)\t \, \dd s,
\end{align}
where $\Phi_{A-BB\t\Pi}(t, s)$ is given by \eqref{eqn:phi-AB-pi} and $\Pi(t)$ is given by \eqref{sol:pi}. 
In particular, for any $s \in [0, 1]$, as $\Pi(s) \to N(1, s)^{-1}$, then $\Sigma(1) \to 0_{n \times n}$. 
When $\Sigma_{0} \succ 0$, as $\Pi(s) \to N(0, s)^{-1}$, then $\Sigma(1) \to +\infty$.
\end{lemma}

\begin{proof} 
It is not difficult to check that \eqref{sol:sigma} follows directly from \eqref{ode:sigma}. 
By \eqref{sol:pi} and the continuity of $\Pi$, 
it follows that, as $\Pi(s) \to N(1, s)^{-1}$, 
then the corresponding solutions of the Riccati equations for $\Pi(t)$ and $N(1, t)^{-1}$ from \eqref{ode:pi} and \eqref{ode:RiccN}, respectively, 
satisfy $\Pi(t) \to N(1, t)^{-1}$ for each $t \in [0, 1]$. 
By the dominated convergence theorem, it follows that, as $\Pi(t) \to N(1, t)^{-1}$ for all $t \in [0, 1]$, then $\Sigma(1) \to 0_{n \times n}$. 
Similarly, as $\Pi(s) \to N(0, s)^{-1}$, then $\Pi(t) \to N(0, t)^{-1}$ for each $t \in [0, 1]$. 
When $\Sigma_{0} \succ 0$, as $\Pi(0) \to N(0, 0^{+})^{-1} = -\infty$, then $\Sigma(1) \to +\infty$. \hfill
\end{proof}

For the special case when $\Sigma_{0} \succ 0$ and $C(t) D(t) C\t(t) = B(t) B\t(t)$ for all $t \in [0, 1]$, we have that
\begin{align*}
&\int_{0}^{1} \Phi_{A-BB\t\Pi}(1, \tau) B(\tau) B\t(\tau) \Phi_{A-BB\t\Pi}(1, \tau)\t \, \dd\tau \\
&= \int_{0}^{1} \big(\Phi_{A}(1, \tau) - \Phi_{A}(1, \tau) N(1, \tau) \Pi(\tau)\big) B(\tau) B\t(\tau) \\* 
&\qquad\qquad \times \big( \Phi_{A}(1, \tau)\t - \Pi(\tau) N(1, \tau) \Phi_{A}(1, \tau)\t \big) \, \dd \tau \\
&= - N(0, 1) + \Phi_{A}(1, t) N(1, t) \Pi(t) N(1, t) \Phi_{A}(1, t)\t \big|_{0}^{1} \\
&= - N(0, 1) - \Phi_{A}(1, 0) N(1, 0) \Pi(0) N(1, 0) \Phi_{A}(1, 0)\t.
\end{align*}
For notational simplicity, in the sequel let $N(1, 0)$ be denoted by $N_{10}$, and let $\Pi(0)$ be denoted by $\Pi_{0}$. 
Then,
\begin{align*}
&\Phi_{A}(0, 1) \Sigma(1) \Phi_{A}(0, 1)\t 
= N_{10}  (N_{10}^{-1} - \Pi_{0} ) N_{10} \\* 
& \hspace{28mm} + N_{10} (N_{10}^{-1} - \Pi_{0}) \Sigma_{0} (N_{10}^{-1} - \Pi_{0}) N_{10} \\
& \hspace{3mm} = N_{10} \Sigma_{0}^{-\frac{1}{2}} \bigg(\Sigma_{0}^{\frac{1}{2}} (N_{10}^{-1} - \Pi_{0}) \Sigma_{0} (N_{10}^{-1} - \Pi_{0}) 
\Sigma_{0}^{\frac{1}{2}} \\* 
& \hspace{28mm} + \Sigma_{0}^{\frac{1}{2}} (N_{10}^{-1} - \Pi_{0}) \Sigma_{0}^{\frac{1}{2}}\bigg) \Sigma_{0}^{-\frac{1}{2}} N_{10} \\
& \hspace{3mm} = N_{10} \Sigma_{0}^{-\frac{1}{2}} \! \left(\left[\Sigma_{0}^{\frac{1}{2}} (N_{10}^{-1} - \Pi_{0}) 
\Sigma_{0}^{\frac{1}{2}} + \frac{I}{2}\right]^{2} - \frac{I}{4}\right) \! \Sigma_{0}^{-\frac{1}{2}} N_{10}.
\end{align*}
Therefore, given $\Sigma(1) = \Sigma_{1} \succ 0$, $\Pi(0)$ is unique and is given by
\begin{align*}
& \Pi(0) = \frac{\Sigma_{0}^{-1}}{2} + N_{10}^{-1} \\* 
& - \Sigma_{0}^{-\frac{1}{2}} \! \left(\frac{I}{4} + \Sigma_{0}^{\frac{1}{2}} N_{10}^{-1} \Phi_{A}(0, 1) \Sigma_{1} \Phi_{A}(0, 1)\t N_{10}^{-1} \Sigma_{0}^{\frac{1}{2}}\right)^{\frac{1}{2}} \! \Sigma_{0}^{-\frac{1}{2}}.
\end{align*}
This is the same solution reported in \cite{chen2016I}.

\subsection{Map Between Boundary Values of the Coupled ODEs}

Since $\Phi_{A-BB\t\Pi}(1, s) = \Phi_{A-BB\t\Pi}(1, 0) \Phi_{A-BB\t\Pi}(0, s)$, Lemma \ref{lem:sigma1} provides an explicit map from $\Pi(0)$ to $\Sigma(1)$. 
Specifically, define $f: \big\{\Pi_{0} \in \R^{n \times n} \,|\, \Pi_{0} = \Pi_{0}\t \prec N_{10}^{-1}\big\} \to \big\{\Sigma_{1} \in \R^{n \times n} \,|\, \Sigma_{1} = \Sigma_{1}\t \succ 0\big\}$ such that
\begin{align} \label{map:pi0-sigma1}
&f(\Pi_{0}) =  \Phi_{A_{10}} (I - N_{10}\Pi_{0}) \bigg[ \Sigma_{0} ~ + \nonumber \\* 
&\int_{0}^{1} (I - N_{s0}\Pi_{0})^{-1} \Phi_{A_{0s}} C_{s} D_{s} C_{s}\t \Phi_{A_{0s}}\t (I - \Pi_{0}N_{s0})^{-1} \dd s \bigg] \nonumber \\* 
&\times (I - \Pi_{0}N_{10}) \Phi_{A_{10}}\t,
\end{align}
where $\Phi_{A_{ts}} \triangleq \Phi_{A}(t, s)$, $N_{ts} \triangleq N(t, s)$, $C_{s} \triangleq C(s)$, and $D_{s} \triangleq D(s)$. 
Lemma \ref{lem:sigma1} leads naturally to an alternative sufficient condition for optimality stated below.

\begin{corollary} \label{cor:suff-sigma1}
Assume there exists an $n \times n$ symmetric matrix $\Pi_{0} \prec N_{10}^{-1}$ such that the terminal state covariance $\Sigma_{1}$ can be written as $\Sigma_{1} = f(\Pi_{0})$. 
Then, the optimal control is 
\begin{equation*}
u^{*}(t) = - B\t(t) \Pi(t) x(t),
\end{equation*}
where $\Pi(t)$ is the unique solution to \eqref{ode:pi} with $\Pi(0) = \Pi_{0}$.
\end{corollary}

\subsection{Existence and Uniqueness of the Optimal Control}

We are ready to show that the map \eqref{map:pi0-sigma1} defines a one-to-one correspondence between $\Pi(0)$ and $\Sigma(1)$, which leads to the existence and uniqueness of the optimal control. 
As usual, $\big(A(t), B(t)\big)$ is assumed to be totally controllable on $[0, 1]$.

For the sake of notational simplicity, let
\begin{equation*}
\Phi_{\Pi_{10}} \triangleq \Phi_{A-BB\t\Pi}(1, 0) = \Phi_{A_{10}} (I - N_{10}\Pi_{0}),
\end{equation*}
and
\begin{align*}
T_{s0} \triangleq 
\begin{cases}
(N_{s0}^{-1} - \Pi_{0})^{-1} \succ 0, & s \in (0, 1], \\* 
0_{n \times n}, & s = 0.
\end{cases}
\end{align*}


The vectorization $\vect(H)$ of an $n \times n$ matrix $H = [h_{ij}]$ is 
\begin{equation*}
\vect(H) \triangleq [h_{11} \enspace \dots \enspace h_{n1} \enspace h_{12} \enspace \dots \enspace h_{n2} \enspace \dots \enspace h_{1n} \enspace \dots \enspace h_{nn}]\t. 
\end{equation*}
Define the map $\bar{f}: \big\{\vect(\Pi_{0}) \in \R^{n^{2}} \,\big|\, \Pi_{0} = \Pi_{0}\t \prec N_{10}^{-1}\big\} \to \big\{\vect(\Sigma_{1}) \in \R^{n^{2}} \,\big|\, \Sigma_{1} = \Sigma_{1}\t \succ 0\big\}$ such that 
\begin{equation} \label{map:pi0-sigma1-vec}
\bar{f}\big(\vect(\Pi_{0})\big) = \vect\big(f(\Pi_{0})\big),
\end{equation}
where $f$ is defined by \eqref{map:pi0-sigma1}.


\begin{lemma} \label{lem:map-jacob}
The Jacobian of the map $\bar{f}$, defined by \eqref{map:pi0-sigma1-vec}, at $\vect(\Pi_{0})$ is given by
\begin{multline} \label{eqn:map-jacob}
\partial {\bar{f}} \big(\vect(\Pi_{0})\big) = - \Phi_{\Pi_{10}} \otimes \Phi_{\Pi_{10}} \bigg[ \Sigma_{0} \otimes T_{10} + T_{10} \otimes \Sigma_{0} \\* 
+ \int_{0}^{1} P_{s} \otimes \big(T_{10} - T_{s0}\big) + \big(T_{10} - T_{s0}\big) \otimes P_{s} \, \dd s \bigg],
\end{multline}
where $\otimes$ denotes the Kronecker product and 
\begin{equation*}
P_{s} \triangleq (I - N_{s0}\Pi_{0})^{-1} \Phi_{A_{0s}} C_{s} D_{s} C_{s}\t \Phi_{A_{0s}}\t (I - \Pi_{0}N_{s0})^{-1} \succeq 0.
\end{equation*}
\end{lemma}

\begin{proof}
See the Appendix. \hfill
\end{proof}


\begin{lemma} \label{lem:map-inv}
For any given $\Sigma_{0} \succ 0$, the map $f$ defined by \eqref{map:pi0-sigma1} is a homeomorphism. 
Thus, for any $\Sigma_{1} \succ 0$, there exists a unique $\Pi_{0} \prec N_{10}^{-1}$ such that $\Sigma_{1} = f(\Pi_{0})$.
\end{lemma}

\begin{proof}
See the Appendix. \hfill
\end{proof}

To this end, the main result of the paper, Theorem \ref{thm:exist-unique}, follows immediately from Lemma \ref{lem:map-inv} and Corollary \ref{cor:suff-sigma1}.


In general, it is not straightforward to get an analytic expression of the inverse map $f^{-1}$ from $\Sigma(1)$ to $\Pi(0)$. 
Instead, we have to resort to numerical methods. 
Specifically, 
since the map $f$ from $\Pi(0)$ to $\Sigma(1)$ and its Jacobian are known, 
we can use root-finding algorithms to find $\Pi(0)$ for given $\Sigma(1)$. 
When $C(t) D(t) C\t(t)$ is approximately equal to a scalar multiple of $B(t) R^{-1}(t) B\t(t)$, 
one can adopt perturbation methods to find an approximate solution to the coupled ODEs. 
Lastly, as explained in \cite{chen2016II}, the general optimal covariance steering problem can always be solved numerically by recasting it as a semidefinite program.


When $n > 1$, the map $f$ from $\Pi(0)$ to $\Sigma(1)$ is not monotone in the Loewner order\footnote{Loewner order is a partial order on the set of $n \times n$ symmetric matrices in that, for two such matrices $M_{1}$ and $M_{2}$, $M_{1} \geq M_{2}$ if $M_{1} \succeq M_{2}$, and $M_{1} > M_{2}$ if $M_{1} \succ M_{2}$.}. 
However, when $n = 1$, $f$ is monotonically decreasing. 
Thus, we can find $\Pi(0)$ using, for example, bisection. 
In addition, for $n=1$ we can also handle the pathological case of a singular initial covariance.

\subsection{One-Dimensional Covariance Control with Singular Initial Covariance}

In this subsection, we explain how to optimally steer the state covariance when the initial covariance is zero, that is, the initial state is deterministic, by exploiting the monotonicity of $f$ when $n = 1$. 
Without loss of generality, assume $p = q = 1$. 
Let the one-dimensional variables be denoted by the lower-case letters of the corresponding matrix variables. 
With $n = 1$, the property of total controllability reduces to the property that, for all $0 \leq t_{0} < t_{1} \leq 1$, there exists $t \in (t_{0}, t_{1})$ such that $b(t) \neq 0$, which is assumed in this subsection.

When $n=1$, let $\sigma(t; \pi_{0})$ denote the state covariance \eqref{sol:sigma} starting from $\pi(0) = \pi_{0} < N(1, 0)^{-1}$ and $\sigma(0) = \sigma_{0} \geq 0$. 
It follows from the monotonicity of the Riccati differential equation \cite{freiling1996generalized} that $\sigma(1; \pi_{0})$ is continuous and monotonically decreasing in $\pi_{0}$. 
The following result shows that with a singular initial state covariance $\sigma(0) = 0$, the desired terminal covariance may not be too large, depending on the initial noise.

\begin{theorem} \label{thm:exist-unique-1d}
Let $\sigma_{0} = 0$. 
For any given $\sigma_{1} \in (0, \eta)$, where 
\begin{align*}
&\eta \triangleq \\* 
&\hspace{-0.2mm} \int_{0}^{1} \hspace{-1.5mm} \left( \hspace{-0.5mm} \frac{\int_{0}^{1} e^{\int_{\tau}^{1} 2a(s) \, \dd s} r^{-1}(\tau) b^{2}(\tau) \, \dd\tau}{\int_{0}^{t} e^{\int_{\tau}^{t} 2a(s) \, \dd s} r^{-1}(\tau) b^{2}(\tau) \, \dd\tau} \hspace{-0.5mm} \right)^{2} \hspace{-2mm} e^{- \int_{t}^{1} 2a(\tau) \, \dd\tau} d(t) c^{2}(t) \, \dd t,
\end{align*}
the unique optimal control for solving the covariance steering problem is given by \eqref{ctrl:opt}. 
As $\pi(0) \to -\infty$, $\sigma(1) \to \eta^{-}$. 
Moreover\footnote{If $c(0)d(0) = 0$, the value of $\eta$ depends on the asymptotic behavior of $d(t) c^{2}(t)$ as $t \to 0^{+}$.}, if $c(0)d(0) \neq 0$, $\eta = +\infty$. 
\end{theorem}

\begin{proof}
See the Appendix. \hfill
\end{proof}

\section{Numerical Examples} \label{sec:expl}

We illustrate the results of the theory using two examples. 
The first example is a simple problem of controlled population growth. 
Assume the population is susceptible to environmental uncertainties such as epidemics. 
Let $x \geq 0$ be the population under control. 
The effects of epidemics on the population can be modeled by the negative of a nonhomogeneous Poisson process $h(t)$ with arrival rate $\lambda(t)$~\cite{o2021comparative}. 
The controlled population subject to environmental uncertainties is modeled by the linear stochastic differential equation
\begin{equation*}
\dd x(t) = a(t) x(t) \, \dd t + u(t) \, \dd t - 4 \, \dd h(t) + 2 \, \dd w(t), \hspace{2mm} t \in [0, 1].
\end{equation*}
The equation can be written, equivalently, as
\begin{equation*}
\dd x(t) = a(t) x(t) \, \dd t - 4\lambda(t) \, \dd t + u(t) \, \dd t + \dd m(t), \hspace{2mm} t \in [0, 1],
\end{equation*}
where $\dd m(t) = - 4\big(\dd h(t) - \lambda(t) \, \dd t\big) + 2 \, \dd w(t)$. Assume that $a(t) = 0.8 - 0.1t$ and $\lambda(t) = 2 + t$. 
The goal is to control the population over the time interval $[0, 1]$ from an initial distribution with mean $\mu_{0} = 50$ and variance $\sigma_0 = 6$ to a target distribution with mean $\mu_{1} = 60$ and variance $\sigma_{1} = 2$ using the least control energy $\mathbb{E}\left[\int_{0}^{1} u^{2}(t) \, \dd t\right]$. 
The optimal control $u^{*}$ is given by
\begin{equation*}
u^{*}(t) = - \pi(t)x(t) + 4\lambda(t) + \nu(t),
\end{equation*}
where $\pi(t)$ is the unique solution to
\begin{align*}
\dot{\pi} &= - 2a(t)\pi + \pi^{2}, \\
\dot{\sigma} &= 2\big(a(t) - \pi(t)\big)\sigma + 16\lambda(t) + 4, \enspace \sigma(0) = 6, \enspace \sigma(1) = 2,
\end{align*}
and 
\begin{align*}
&\nu(t) = \phi(1, t)\left(\int_{0}^{1} \phi(1, \tau)^{2} \, \dd\tau\right)^{-1} \hspace{-1mm} \big(\mu_{1} - \phi(1, 0)\mu_{0} \big), \\
&\dot{\phi}(1, t) = \big(\pi(t) - a(t)\big) \phi(1, t), \quad \phi(1, 1) = 1.
\end{align*}
Assuming $x(0) \sim \mathcal{N}(\mu_{0}, \sigma_{0})$, ten sample paths of $x(t)$ using the optimal control $u^{*}$ are plotted in Figure \ref{fig:expl-1}.
The transparent blue region is the three-standard deviation interval between $\mathbb{E}\left[x(t)\right] - 3\sqrt{\sigma}(t)$ and $\mathbb{E}\left[x(t)\right] + 3\sqrt{\sigma}(t)$, $t \in [0, 1]$.
\begin{figure}[!h]
    \centering
    \includegraphics[width=0.48\textwidth]{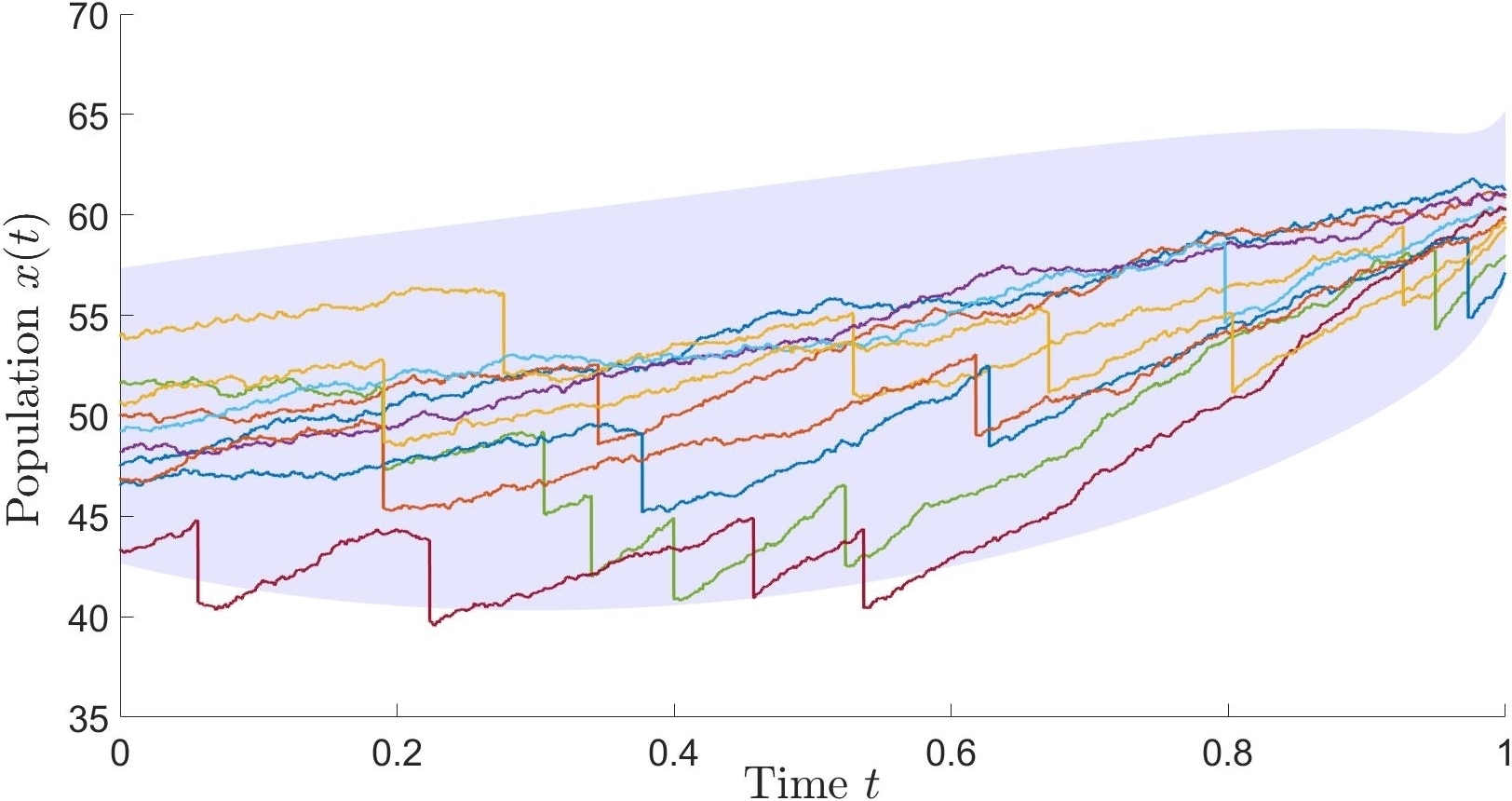}
    \caption{Controlled population: $10$ sample paths.}
    \label{fig:expl-1}
\end{figure}


In the second example, we consider a small flying vehicle in heavy rain, which is modeled by a kinematic point subject to Poisson and Gaussian noise~\cite{cowpertwait2007point}. 
We only consider the motion of the vehicle along the vertical direction. 
Let $x(t)$ and $v(t)$ denote the vertical position and velocity, respectively. 
The effect of rain on the vehicle is modeled by a nonhomogeneous compound Poisson process $h(t)$. 
The vertical dynamics of the vehicle is therefore given by
\begin{align*}
\dd x(t) &= v(t) \, \dd t, \\
\dd v(t) &= u(t) \, \dd t + 0.2 \, \dd w(t) + \dd h(t),
\end{align*}
where $h(t)$ has arrival rate $\lambda(t) = 5 - t$ and i.i.d. jump size $\chi \sim \mathcal{N}(-0.5, 0.1^{2})$. 
Our goal is to control the covariance of the vehicle while hovering at a certain height from an initial covariance of $\Sigma_{0} = \diag [0.6, 0.6]$ to a target covariance of $\Sigma_{1} = \diag [0.2, 0.1]$ with the least effort $\mathbb{E}\left[\int_{0}^{1} u^{2}(t) \, \dd t\right]$. 
Assume $\left[x(0) \enspace v(0)\right]\t \sim \mathcal{N}\big(\left[0 \enspace 0\right]\t, \Sigma_{0}\big)$. 
Figure \ref{fig:expl-2} illustrates ten controlled sample paths in the phase space as a function of time. 
The three-standard deviation tolerance interval for $t \in [0, 1]$ is depicted as the transparent tube.
\begin{figure}[!h]
    \centering
    \includegraphics[width=0.48\textwidth]{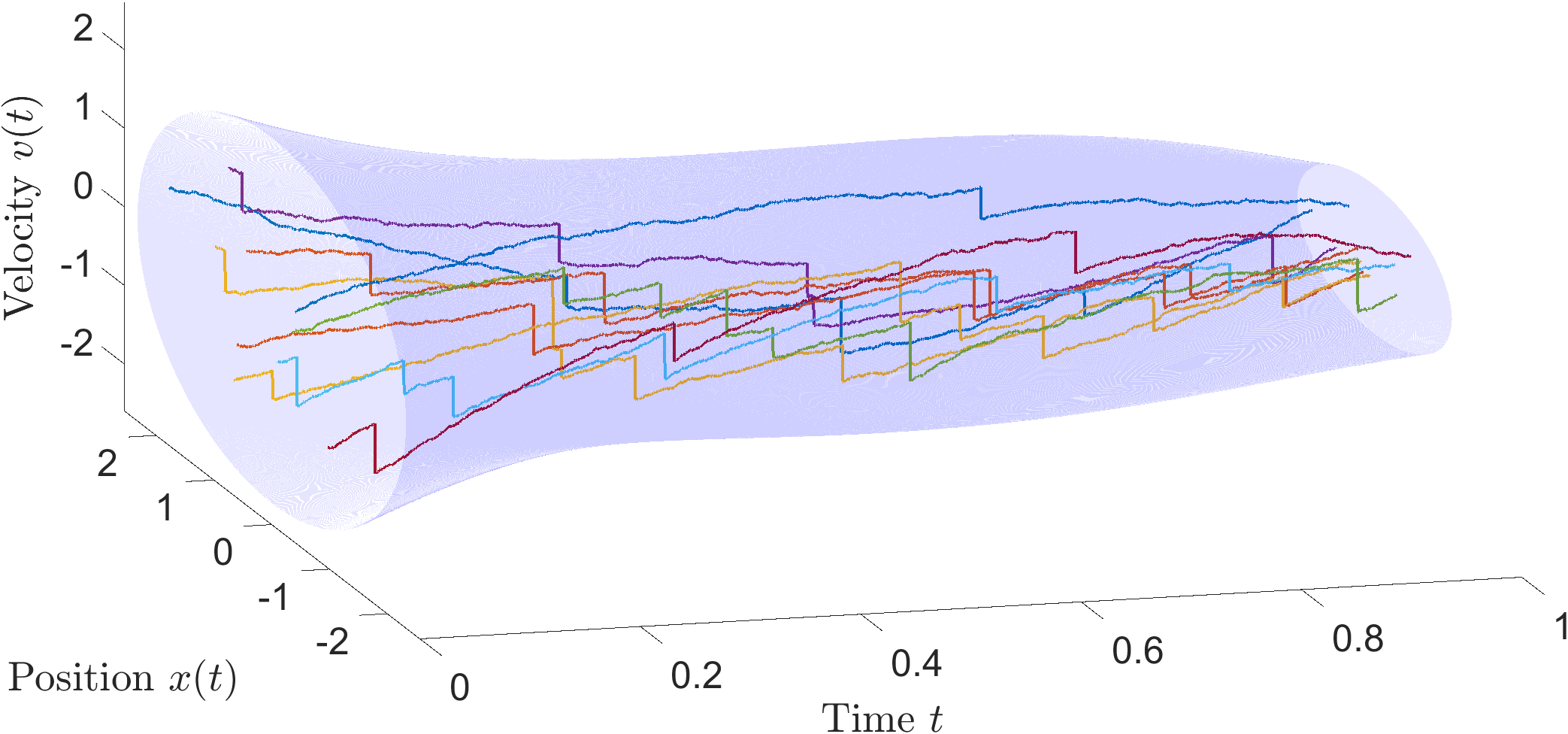}
    \caption{Controlled position and velocity: $10$ sample paths.}
    \label{fig:expl-2}
\end{figure}

\section{Concluding Remarks}

In this paper, the optimal control law is established for steering the state covariance of a general linear time-varying stochastic system subject to additive noise. 
%
Interesting potential extensions include the development of efficient algorithms for solving coupled matrix ODEs \eqref{ode:pi}, \eqref{ode:sigma}, \eqref{bdr:sigma} and the optimal control of the state covariance along the trajectory, or with partially fixed terminal covariance and/or free final time.



\bibliographystyle{IEEEtran}
\bibliography{TAC-conv-contr-add-2column-v3}

\begin{thebibliography}{10}
\providecommand{\url}[1]{#1}
\csname url@samestyle\endcsname
\providecommand{\newblock}{\relax}
\providecommand{\bibinfo}[2]{#2}
\providecommand{\BIBentrySTDinterwordspacing}{\spaceskip=0pt\relax}
\providecommand{\BIBentryALTinterwordstretchfactor}{4}
\providecommand{\BIBentryALTinterwordspacing}{\spaceskip=\fontdimen2\font plus
\BIBentryALTinterwordstretchfactor\fontdimen3\font minus
  \fontdimen4\font\relax}
\providecommand{\BIBforeignlanguage}[2]{{%
\expandafter\ifx\csname l@#1\endcsname\relax
\typeout{** WARNING: IEEEtran.bst: No hyphenation pattern has been}%
\typeout{** loaded for the language `#1'. Using the pattern for}%
\typeout{** the default language instead.}%
\else
\language=\csname l@#1\endcsname
\fi
#2}}
\providecommand{\BIBdecl}{\relax}
\BIBdecl

\bibitem{hotz1985covariance}
A.~F. Hotz and R.~E. Skelton, ``A covariance control theory,'' in \emph{Proc.
  {IEEE} Conf. Decision Control}, Lauderdale, FL, 1985, pp. 552--557.

\bibitem{collins1985covariance}
E.~Collins and R.~Skelton, ``Covariance control of discrete systems,'' in
  \emph{Proc. {IEEE} Conf. Decision Control}, Lauderdale, FL, 1985, pp.
  542--547.

\bibitem{ridderhof2018uncertainty}
J.~Ridderhof and P.~Tsiotras, ``Uncertainty quantication and control during
  {M}ars powered descent and landing using covariance steering,'' in \emph{AIAA
  Guidance, Navigation, Control Conf.}, Kissimmee, FL, 2018.

\bibitem{ridderhof2020fuel}
J.~Ridderhof, J.~Pilipovsky, and P.~Tsiotras, ``Chance-constrained covariance
  control for low-thrust minimum-fuel trajectory optimization,'' in
  \emph{AAS/AIAA Astrodynamics Specialist Conf.}, South Lake Tahoe, CA, 2020.

\bibitem{okamoto2019optimal}
K.~Okamoto and P.~Tsiotras, ``Optimal stochastic vehicle path planning using
  covariance steering,'' \emph{IEEE Robot. Autom. Lett.}, vol.~4, no.~3, pp.
  2276--2281, 2019.

\bibitem{zheng2022belief}
D.~Zheng, J.~Ridderhof, P.~Tsiotras, and A.-a. Agha-mohammadi, ``Belief space
  planning: a covariance steering approach,'' in \emph{Int. Conf. Robot.
  Autom.}, Philadelphia, PA, 2022, pp. 11\,051--11\,057.

\bibitem{yin2022trajectory}
J.~Yin, Z.~Zhang, E.~Theodorou, and P.~Tsiotras, ``Trajectory distribution
  control for model predictive path integral control using covariance
  steering,'' in \emph{Int. Conf. Robot. Autom.}, Philadelphia, PA, 2022, pp.
  1478--1484.

\bibitem{saravanos2021distributed}
A.~D. Saravanos, A.~G. Tsolovikos, E.~Bakolas, and E.~A. Theodorou,
  ``Distributed covariance steering with consensus {ADMM} for stochastic
  multi-agent systems,'' in \emph{Robot.: Sci. Syst.}, 2021.

\bibitem{chen2021optimal}
Y.~Chen, T.~T. Georgiou, and M.~Pavon, ``Optimal transport in systems and
  control,'' \emph{Annu. Rev. Control Robot. Auton. Syst.}, vol.~4, pp.
  89--113, 2021.

\bibitem{chen2016stochastic}
Y.~Chen and T.~T. Georgiou, ``Stochastic bridges of linear systems,''
  \emph{{IEEE} Trans. Autom. Control}, vol.~61, no.~2, pp. 526--531, 2016.

\bibitem{chen2016relation}
Y.~Chen, T.~T. Georgiou, and M.~Pavon, ``On the relation between optimal
  transport and {S}chr{\"o}dinger bridges: a stochastic control viewpoint,''
  \emph{J. Optim. Theory Appl.}, vol. 169, pp. 671--691, 2016.

\bibitem{hotz1987covariance}
A.~Hotz and R.~E. Skelton, ``Covariance control theory,'' \emph{Int. J.
  Control}, vol.~46, no.~1, pp. 13--32, 1987.

\bibitem{collins1987theory}
E.~Collins and R.~Skelton, ``A theory of state covariance assignment for
  discrete systems,'' \emph{{IEEE} Trans. Autom. Control}, vol.~32, no.~1, pp.
  35--41, 1987.

\bibitem{yasuda1993covariance}
K.~Yasuda, R.~E. Skelton, and K.~M. Grigoriadis, ``Covariance controllers: a
  new parametrization of the class of all stabilizing controllers,''
  \emph{Automatica}, vol.~29, no.~3, pp. 785--788, 1993.

\bibitem{georgiou2002structure}
T.~T. Georgiou, ``The structure of state covariances and its relation to the
  power spectrum of the input,'' \emph{{IEEE} Trans. Autom. Control}, vol.~47,
  no.~7, pp. 1056--1066, 2002.

\bibitem{zhu1997convergent}
G.~Zhu, M.~Rotea, and R.~Skelton, ``A convergent algorithm for the output
  covariance constraint control problem,'' \emph{{SIAM} J. Control Optim.},
  vol.~35, no.~1, pp. 341--361, 1997.

\bibitem{chen2016I}
Y.~Chen, T.~T. Georgiou, and M.~Pavon, ``Optimal steering of a linear
  stochastic system to a final probability distribution, part {I},''
  \emph{{IEEE} Trans. Autom. Control}, vol.~61, no.~5, pp. 1158--1169, 2016.

\bibitem{chen2016II}
Y.~Chen, T.~T. Georgiou, and M.~Pavon, ``Optimal steering of a linear
  stochastic system to a final probability distribution, part {II},''
  \emph{{IEEE} Trans. Autom. Control}, vol.~61, no.~5, pp. 1170--1180, 2016.

\bibitem{chen2018III}
Y.~Chen, T.~T. Georgiou, and M.~Pavon, ``Optimal steering of a linear
  stochastic system to a final probability distribution, part {III},''
  \emph{{IEEE} Trans. Autom. Control}, vol.~63, no.~9, pp. 3112--3118, 2018.

\bibitem{ciccone2020regularized}
V.~Ciccone, Y.~Chen, T.~T. Georgiou, and M.~Pavon, ``Regularized transport
  between singular covariance matrices,'' \emph{{IEEE} Trans. Autom. Control},
  vol.~66, no.~7, pp. 3339--3346, 2020.

\bibitem{bakolas2018finite}
E.~Bakolas, ``Finite-horizon covariance control for discrete-time stochastic
  linear systems subject to input constraints,'' \emph{Automatica}, vol.~91,
  pp. 61--68, 2018.

\bibitem{balci2021covariance}
I.~M. Balci and E.~Bakolas, ``Covariance control of discrete-time {G}aussian
  linear systems using affine disturbance feedback control policies,'' in
  \emph{Proc. {IEEE} Conf. Decision Control}, Austin, TX, 2021, pp. 2324--2329.

\bibitem{balci2021Wasserstein}
I.~M. Balci and E.~Bakolas, ``Covariance steering of discrete-time stochastic
  linear systems based on {W}asserstein distance terminal cost,'' \emph{IEEE
  Control Syst. Lett.}, vol.~5, no.~6, pp. 2000--2005, 2021.

\bibitem{okamoto2018optimal}
K.~Okamoto, M.~Goldshtein, and P.~Tsiotras, ``Optimal covariance control for
  stochastic systems under chance constraints,'' \emph{IEEE Control Syst.
  Lett.}, vol.~2, no.~2, pp. 266--271, 2018.

\bibitem{sivaramakrishnan2021distribution}
V.~Sivaramakrishnan, J.~Pilipovsky, M.~M. Oishi, and P.~Tsiotras,
  ``Distribution steering for discrete-time linear systems with general
  disturbances using characteristic functions,'' in \emph{Proc. Amer. Control
  Conf.}, Atlanta, GA, 2022.

\bibitem{pilipovsky2021covariance}
J.~Pilipovsky and P.~Tsiotras, ``Covariance steering with optimal risk
  allocation,'' \emph{IEEE Trans. Aerosp. Electron. Syst.}, vol.~57, no.~6, pp.
  3719--3733, 2021.

\bibitem{cinlar2011stochastics}
E.~{\c{C}}inlar, \emph{Probability and Stochastics}.\hskip 1em plus 0.5em minus
  0.4em\relax Springer, 2011, vol. 261.

\bibitem{protter2003stochastic}
P.~E. Protter, \emph{Stochastic Integration and Differential Equations}.\hskip
  1em plus 0.5em minus 0.4em\relax Springer, 2003, vol.~21.

\bibitem{silverman1967controllability}
L.~M. Silverman and H.~Meadows, ``Controllability and observability in
  time-variable linear systems,'' \emph{{SIAM} J. Control}, vol.~5, no.~1, pp.
  64--73, 1967.

\bibitem{kreindler1964concepts}
E.~Kreindler and P.~Sarachik, ``On the concepts of controllability and
  observability of linear systems,'' \emph{{IEEE} Trans. Autom. Control},
  vol.~9, no.~2, pp. 129--136, 1964.

\bibitem{stubberud1964controllability}
A.~Stubberud, ``A controllability criterion for a class of linear systems,''
  \emph{{IEEE} Trans. Ind. Appl.}, vol.~83, no.~75, pp. 411--413, 1964.

\bibitem{morse1973structure}
A.~Morse and L.~Silverman, ``Structure of index-invariant systems,''
  \emph{{SIAM} J. Control}, vol.~11, no.~2, pp. 215--225, 1973.

\bibitem{wolovich1968stabilization}
W.~Wolovich, ``On the stabilization of controllable systems,'' \emph{{IEEE}
  Trans. Autom. Control}, vol.~13, no.~5, pp. 569--572, 1968.

\bibitem{seal1969canonical}
C.~Seal and A.~Stubberud, ``Canonical forms for multiple-input time-variable
  systems,'' \emph{{IEEE} Trans. Autom. Control}, vol.~14, no.~6, pp. 704--707,
  1969.

\bibitem{kilicaslan2010existence}
S.~Kilicaslan and S.~P. Banks, ``Existence of solutions of {R}iccati
  differential equations for linear time varying systems,'' in \emph{Proc.
  Amer. Control Conf.}, Baltimore, MD, 2010, pp. 1586--1590.

\bibitem{freiling1996generalized}
G.~Freiling, G.~Jank, and H.~Abou-Kandil, ``Generalized {R}iccati difference
  and differential equations,'' \emph{Linear Algebra Its Appl.}, vol. 241, pp.
  291--303, 1996.

\bibitem{o2021comparative}
M.~O’Driscoll, C.~Harry, C.~A. Donnelly, A.~Cori, and I.~Dorigatti, ``A
  comparative analysis of statistical methods to estimate the reproduction
  number in emerging epidemics, with implications for the current coronavirus
  disease 2019 ({COVID}-19) pandemic,'' \emph{Clin. Infect. Dis.}, vol.~73,
  no.~1, pp. e215--e223, 2021.

\bibitem{cowpertwait2007point}
P.~Cowpertwait, V.~Isham, and C.~Onof, ``Point process models of rainfall:
  developments for fine-scale structure,'' \emph{Proc. R. Soc. A: Math. Phys.
  Eng. Sci.}, vol. 463, no. 2086, pp. 2569--2587, 2007.

\bibitem{henderson1981deriving}
H.~V. Henderson and S.~R. Searle, ``On deriving the inverse of a sum of
  matrices,'' \emph{SIAM Rev.}, vol.~23, no.~1, pp. 53--60, 1981.

\bibitem{krantz2002implicit}
S.~G. Krantz and H.~R. Parks, \emph{The Implicit Function Theorem: History,
  Theory, and Applications}.\hskip 1em plus 0.5em minus 0.4em\relax Springer,
  2002.

\end{thebibliography}


\section*{Appendix}

\begin{proof}[Proof of Lemma~\ref{lem:map-jacob}]
We can compute the Jacobian of the map $\bar{f}$ defined by \eqref{map:pi0-sigma1-vec} as follows. 
Let $\Delta {\Pi}_{0}$ denote a small increment of $\Pi_{0}$. Then, it follows from~\cite{henderson1981deriving} that
\begin{align*}
&(I - N_{s0} \Pi_{0} - N_{s0} \Delta {\Pi}_{0})^{-1} = (I - N_{s0}\Pi_{0})^{-1} \\* 
&\hspace{1mm} + (I - N_{s0} \Pi_{0})^{-1} N_{s0} \Delta {\Pi}_{0} 
(I - N_{s0} \Pi_{0})^{-1} \hspace{-1mm} + O\big(\|\Delta {\Pi}_{0}\|^{2}\big).
\end{align*}
Hence, by collecting all the first order terms of $\Delta {\Pi}_{0}$, we have
\begin{align*}
&f\big(\Pi_{0} + \Delta {\Pi}_{0}\big) - f\big(\Pi_{0}\big) = O\big(\|\Delta {\Pi}_{0}\|^{2}\big) \\* 
&\hspace{4mm} - \Phi_{A_{10}} N_{10} \Delta {\Pi}_{0} \bigg[\Sigma_{0} + \int_{0}^{1} P_{s} \, \dd s\bigg] (I - \Pi_{0}N_{10}) \Phi_{A_{10}}\t \\* 
&\hspace{4mm} - \Phi_{A_{10}} (I - N_{10}\Pi_{0}) \bigg[\Sigma_{0} + \int_{0}^{1} P_{s} \, \dd s\bigg] \Delta {\Pi}_{0} N_{10} \Phi_{A_{10}}\t \\* 
&\hspace{4mm} + \Phi_{A_{10}} (I - N_{10} \Pi_{0}) \bigg[\int_{0}^{1} 
(I - N_{s0} \Pi_{0})^{-1} N_{s0} \Delta {\Pi}_{0} P_{s} \\* 
&\hspace{9mm} + P_{s} \Delta {\Pi}_{0} N_{s0} 
(I - \Pi_{0} N_{s0})^{-1} \dd s \bigg] (I - \Pi_{0}N_{10}) \Phi_{A_{10}}\t.
\end{align*}
Thus, we can write
\begin{align*} \label{eqn:map-diff}
&f\big(\Pi_{0} + \Delta {\Pi}_{0}\big) - f\big(\Pi_{0}\big) = O\big(\|\Delta {\Pi}_{0}\|^{2}\big) \nonumber \\* 
&\hspace{2mm} - \Phi_{\Pi_{10}} \bigg[ T_{10} \Delta {\Pi}_{0} \Sigma_{0} + \Sigma_{0} \Delta {\Pi}_{0} T_{10} \nonumber \\* 
&\hspace{2mm} + \int_{0}^{1} \! (T_{10} - T_{s0}) \Delta {\Pi}_{0} P_{s} + P_{s} \Delta {\Pi}_{0} (T_{10} - T_{s0}) \, \dd s \bigg] \Phi_{\Pi_{10}}\t.
\end{align*}
Vectorizing both sides of the previous equation yields 
\begin{multline*}
\vect \Big( f\big(\Pi_{0} + \Delta {\Pi}_{0}\big) - f\big(\Pi_{0}\big) \Big) = \partial {\bar{f}}\big(\vect(\Pi_{0})\big) \vect\big(\Delta {\Pi}_{0}\big) \\* 
+ O\big(\|\Delta {\Pi}_{0}\|^{2}\big).
\end{multline*}
It follows that
\begin{align*} 
\partial {\bar{f}}\big(\vect(\Pi_{0})\big) &= - \Phi_{\Pi_{10}} \otimes 
\Phi_{\Pi_{10}} \bigg[ \Sigma_{0} \otimes T_{10} + T_{10} \otimes \Sigma_{0} \nonumber \\* 
&\hspace{-4mm} + \int_{0}^{1} P_{s} \otimes \big(T_{10} - T_{s0}\big) + \big(T_{10} - T_{s0}\big) \otimes P_{s} \, \dd s \bigg],
\end{align*}
thus completing the proof. \hfill 
\end{proof}


\begin{proof}[Proof of Lemma~\ref{lem:map-inv}]
Since $T_{s0}$ is continuous in $\Pi_{0}$ and $s \in [0, 1]$, $\partial {\bar{f}}\big(\vect(\Pi_{0})\big)$ is continuous in $\vect(\Pi_{0})$. 
First, we show that $\partial {\bar{f}}\big(\vect(\Pi_{0})\big)$ is nonsingular at each $\Pi_{0} \prec N_{10}^{-1}$. 
Since $\Phi_{\Pi_{10}}$ is nonsingular, $\Phi_{\Pi_{10}} \otimes \Phi_{\Pi_{10}}$ is nonsingular as well. 
It suffices to show that the term in the square brackets of \eqref{eqn:map-jacob}, that is,
\begin{multline*}
S \triangleq \Sigma_{0} \otimes T_{10} + T_{10} \otimes \Sigma_{0} \\* 
+ \int_{0}^{1} P_{s} \otimes \big(T_{10} - T_{s0}\big) + \big(T_{10} - T_{s0}\big) \otimes P_{s} \, \dd s,
\end{multline*}
is nonsingular. 
Notice that $S$ is symmetric, since $\Sigma_{0}, T_{10} \succ 0$ and $P_{s}, T_{10} - T_{s0} \succeq 0$ are all symmetric. 
Let $X \neq 0$ be an $n \times n$ matrix, not necessarily symmetric. 
Then,
\begin{align*}
&\vect(X)\t S \vect(X) = \trace \bigg( X\t T_{10} X \Sigma_{0} + X\t \Sigma_{0} X T_{10} \\* 
&\hspace{9mm} + \int_{0}^{1} X\t \big(T_{10} - T_{s0}\big) X P_{s} + X\t P_{s} X \big(T_{10} - T_{s0}\big) \, \dd s \bigg) \\
&\geq \trace \Big( T_{10}^{\frac{1}{2}} X \Sigma_{0} X\t T_{10}^{\frac{1}{2}} + T_{10}^{\frac{1}{2}} X\t \Sigma_{0} X T_{10}^{\frac{1}{2}} \Big) > 0.
\end{align*}
Thus, $S \succ 0$, which implies that $\partial {\bar{f}}\big(\vect(\Pi_{0})\big)$ is nonsingular at each $\Pi_{0} \prec N_{10}^{-1}$.

Next, we show that the map $f$ is proper, that is, for any compact subset $\mathcal{K} \subset \{\Sigma_{1} \in \R^{n \times n} \,|\, \Sigma_{1} = \Sigma_{1}\t \succ 0\}$, the inverse image $f^{-1}(\mathcal{K}) \subset \{\Pi_{0} \in \R^{n \times n} \,|\, \Pi_{0} = \Pi_{0}\t \prec N_{10}^{-1}\}$ is compact. 
Since $f$ is continuous, the inverse image of a closed set is closed. 
Since $\mathcal{K}$ is bounded, in view of \eqref{map:pi0-sigma1}, the set
\begin{equation*}
\Big\{ \Phi_{A_{10}} (I - N_{10}\Pi_{0}) \Sigma_{0} 
(I - \Pi_{0}N_{10}) \Phi_{A_{10}}\t \, \Big| \, \Pi_{0} \in f^{-1}(\mathcal{K}) \Big\}
\end{equation*}
is bounded. 
It follows that $f^{-1}(\mathcal{K})$ is bounded. 
Hence, $f^{-1}(\mathcal{K})$ is compact, and $f$ is proper.

Since the set of positive definite matrices is convex, it is simply connected~\cite{krantz2002implicit}. 
By Hadamard's global inverse function theorem~\cite{krantz2002implicit}, $f$ is a homeomorphism, thus a bijection. \hfill 
\end{proof}


\begin{proof}[Proof of Theorem~\ref{thm:exist-unique-1d}]
By the continuity and monotonicity of $\pi(t; \pi_{0})$ in $\pi_{0}$ for all $t \in [0, 1]$, it follows that, as $\pi_{0} \to -\infty$, $\pi(t; \pi_{0}) \to N(0, t)^{-1}$ monotonically. 
By \eqref{sol:sigma} and the monotone convergence theorem, it follows that
\begin{equation*}
\lim_{\pi_{0} \to -\infty} \sigma(1; \pi_{0}) = \eta^{-}, \quad \text{for}~\sigma_{0} = 0.
\end{equation*}

Let now $c(0)d(0) \neq 0$. 
By continuity, $d(t) c^{2}(t)$ is bounded below from zero on some interval $[0, \varepsilon]$ for $\varepsilon \in (0, 1]$. 
Let $Z > 0$ satisfy
\begin{equation*}
e^{- \int_{t}^{1} 2a(\tau) \, \dd\tau} d(t) c^{2}(t) \geq Z > 0, \quad \forall t \in [0, \varepsilon],
\end{equation*}
and let
\begin{equation*}
\rho \triangleq \int_{0}^{1} e^{\int_{\tau}^{1} 2a(s) \, \dd s} r^{-1}(\tau) b^{2}(\tau) \, \dd\tau > 0.
\end{equation*}
Let $G \in \R$ and $\Omega > 0$ satisfy
\begin{equation*}
2a(t) \leq G, \quad r^{-1}(t) b^{2}(t) \leq \Omega, \quad \forall t \in [0, \varepsilon].
\end{equation*}
Without loss of generality, assume $\varepsilon$ is sufficiently small, such that, for all $t \in [0, \varepsilon]$, $e^{G t} - 1 = G t + O(t^{2})$. 
Then, for all $t \in [0, \varepsilon]$, 
\begin{align*}
\int_{0}^{t} e^{\int_{\tau}^{t} 2a(s) \, \dd s} r^{-1}(\tau) b^{2}(\tau) \, \dd\tau 
&\leq \int_{0}^{t} e^{G(t - \tau)} \Omega \, \dd\tau \\
&= \frac{\Omega\left(e^{G t} - 1\right)}{G} 
\approx \Omega t.
\end{align*}
Therefore, 
\begin{equation*}
\eta \geq \int_{0}^{\varepsilon} \left(\frac{\rho}{\Omega t}\right)^{2} Z \, \dd t = +\infty.
\end{equation*}
This completes the proof. \hfill 
\end{proof}


\end{document}